\documentclass[a4paper,USenglish]{lipics-v2019}
\usepackage[utf8]{inputenc}

\newcommand{\Z}{\mathbb{Z}}
\newcommand{\R}{\mathbb{R}}
\DeclareMathOperator{\m}{m}
\DeclareMathOperator{\comp}{c}
\DeclareMathOperator{\conv}{conv}

\title{Efficient Algorithms for Battleship}

\author{Loïc Crombez}{Université Clermont Auvergne, LIMOS, France \and \url{https://fc.isima.fr/~lcrombez/} }{loic.crombez@uca.fr}{https://orcid.org/0000-0002-9542-5276}{This work has been sponsored by the French government research program ``Investissements d'Avenir'' through the IDEX-ISITE initiative 16-IDEX-0001 (CAP 20-25).
}

\author{Guilherme D. da Fonseca}{Université Aix Marseille, LIS, France \and \url{https://pageperso.lis-lab.fr/guilherme.fonseca/}}{guilherme.fonseca@lis-lab.fr}{https://orcid.org/0000-0002-9807-028X}{This work is supported by the French ANR PRC grant ADDS (ANR-19-CE48-0005) and the Brazilian CAPES-PrInt project number 88881.310248/2018-01.}

\author{Yan Gerard}{Université Clermont Auvergne, LIMOS, France \and   \url{https://yangerard.wordpress.com/} }{yan.gerard@uca.fr}{https://orcid.org/0000-0002-2664-0650}
{This work is supported by the French ANR PRC grant ADDS (ANR-19-CE48-0005).}

\authorrunning{L. Crombez, G. da Fonseca, and Y. Gerard} 

\Copyright{Loïc Crombez, Guilherme da Fonseca, and Yan Gerard} 



\category{} 

\relatedversion{} 

\supplement{}

\nolinenumbers 

\hideLIPIcs  

\EventEditors{Martin Farach-Colton, Giuseppe Prencipe, and Ryuhei Uehara}
\EventNoEds{3}
\EventLongTitle{10th International Conference on Fun with Algorithms (FUN 2020)}
\EventShortTitle{FUN 2020}
\EventAcronym{FUN}
\EventYear{2020}
\EventDate{September 28--30, 2020}
\EventLocation{Favignana Island, Sicily, Italy}
\EventLogo{}
\SeriesVolume{157}
\ArticleNo{11}

\begin{document}

\maketitle

\begin{abstract}
We consider an algorithmic problem inspired by the Battleship game. In the variant of the problem that we investigate, there is a unique ship of shape $S \subset \Z^2$ which has been translated in the lattice $\Z^2$. 
We assume that a player has already hit the ship with a first shot and the goal is to sink the ship using as few shots as possible, that is, by minimizing the number of missed shots.
While the player knows the shape $S$, which position of $S$ has been hit is not known.

Given a shape $S$ of $n$ lattice points, the minimum number of misses that can be achieved in the worst case by any algorithm is called the Battleship complexity of the shape $S$ and denoted $c(S)$.
We prove three bounds on $c(S)$, each considering a different class of shapes. First, we have $c(S) \leq n-1$ for arbitrary shapes and the bound is tight for parallelogram-free shapes. Second, we provide an algorithm that shows that $c(S) = O(\log n)$ if $S$ is an HV-convex polyomino. Third, we provide an algorithm that shows that $c(S) = O(\log \log n)$ if $S$ is a digital convex set. 
This last result is obtained through a novel discrete version of the Blaschke-Lebesgue inequality relating the area and the width of any convex body.
\end{abstract}

\section{Introduction}

We consider a geometric problem inspired by the children's game Battleship. The Wikipedia description of the game is: \begin{quote}
    Battleship (also Battleships or Sea Battle) is a strategy type guessing game for two players. It is played on ruled grids (paper or board) on which each player's fleet of ships (including battleships) are marked. The locations of the fleets are concealed from the other player. Players alternate turns calling ``shots'' at the other player's ships, and the objective of the game is to destroy the opposing player's fleet.
\end{quote} 

\begin{figure}[t]
  \begin{center}
		\includegraphics[height=4cm]{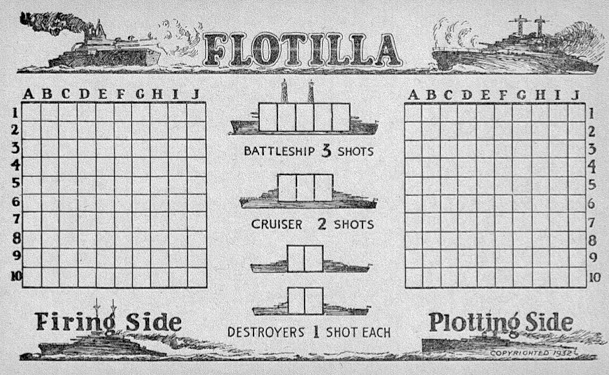}
        \hspace{.5cm}
		\includegraphics[height=4cm]{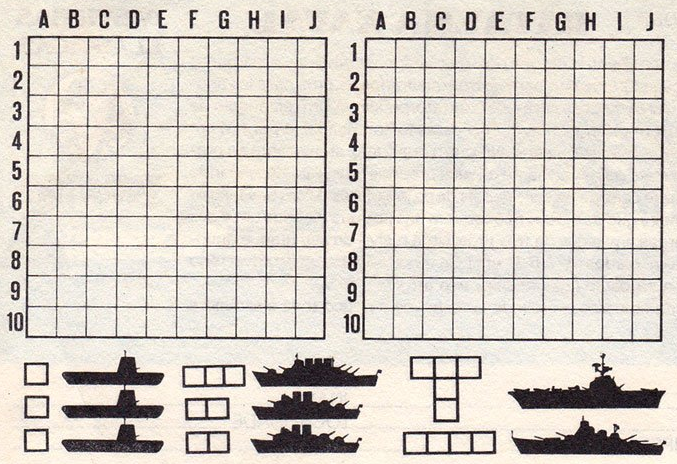}
	\end{center}
	\caption{\label{f:boards}American and Portuguese paper game boards of Battleship.}
\end{figure}

After each shot, the player is informed if the shot has been a ``hit'' or a ``miss'', but no other information is given. In the original version the shapes of the ships are line segments of different lengths. However different countries and commercial brands use a variety of shapes for the ships (see Figure~\ref{f:boards}).

During the game, the strategy of a player (that we call Alice), is usually decomposed in an alternate sequence of two steps:
\begin{enumerate}
\item Hit a new ship of the opponent (that we call Bob).
\item Sink that ship with a minimal number of misses and go back to the first step.
\end{enumerate}

The first step is a hitting set problem and many interesting variations are possible. In this paper, however, we consider the second step.
The goal of the second step is to sink the ship (which has already been hit once) with a minimal number of misses. During a real game, the position of the new ship can be constrained by the positions of the other ships which have been already discovered and the grid boundaries, but we consider a simpler case. The fleet is composed of only one ship placed on an infinite grid using only integer translations. In other words, its shape $S$ is given and can only be translated in the lattice. We know, however the coordinates of one grid cell of the ship.

The second modification to the rules that we make is to forbid rotations. This modification simplifies the problem and since there are at most $4$ possible rotations, the original problem can be solved by considering each rotation separately (at the expense of a factor of $4$).

As a toy example to motivate the problem, we consider the case in which the shape of the ship is a horizontal line segment of length $4$. Alice has already hit Bob's ship, which she knows is a horizontal line segment of length $4$. However, Alice is clueless about which square of the ship she has hit.
In this case, Alice may progressively shoot to the right of the first hit until she misses a shot. At this point, she knows the precise location of Bob's ship and may finish sinking it without missing any additional shot (if she has not already sunk the ship at the fourth shot). In this case, Alice has a strategy that requires at most $1$ missed shot.

We refer to the minimum number of misses that Alice needs to sink a ship of shape $S \subset \Z^2$ as $\comp(S)$. Notice that Alice knows the shape of $S$, but not which square she has initially hit. We just showed that $\comp(S) \leq 1$ if $S$ is a horizontal (or vertical) line segment.
But what happens if the shape $S$ of the ship is not a line segment? 

The goal of this paper is to provide bounds to $\comp(S)$ depending on properties of the shape $S$.
We prove the following results for a shape $S$ of $n$ points:
\begin{itemize}
    \item for arbitrary shapes, $\comp(S) \leq n-1$,
    \item for parallelogram-free shapes, $\comp(S) = n-1$,
    \item for HV-convex polyominoes, $\comp(S) = O(\log n)$, and
    \item for digital convex shapes, $\comp(S) = O(\log \log n)$.
\end{itemize}

The remainder of the paper is organized as follows. Section~\ref{s:prelim} is devoted to formalize our notation and to prove some simple results. In Section~\ref{s:hv}, we provide an algorithm with $O(\log n)$ misses in the worst case for HV-convex polyominoes. In Section~\ref{s:convex}, we present an algorithm with $O(\log \log n)$ misses in the worst case for digital convex sets.
We conclude the paper with a presentation of several open problems and variations.

\section{Preliminaries} \label{s:prelim}

In this section, we formalize our notation and prove some simple results.
Before going further, let us make the notations precise. The ship's \emph{shape} is the finite lattice set $S \subset \Z^2$ and its number of points is $n$. The opponent translated the shape by an unknown vector $-p \in \Z^2$ obtaining a \emph{ship} $S-p$. The vector $p \in \Z^2$ is the \emph{position} of the ship. We say that a \emph{shot} $x$ is a \emph{hit} if $x \in S - p$ and a \emph{miss} otherwise.
By assuming (without loss of generality) that the first hit happens at the origin $x=(0,0)$, we know then that the position $p$ of the ship is a point in the shape $S$ (that is, $p \in S$) but we do not know which point.
In order to determine the actual value of $p$, we are allowed to test the membership in $S$ of points of the form $p+x$ and our goal is to determine $p$ using as few failed membership tests (called \emph{misses}) as possible.

Given a shape $S$, we can model an algorithm to determine the position $p$ by a binary decision tree $T$.
The children of each node correspond to the possible outcomes of the shot: hit or miss. The leaves of the decision tree represent the nodes in which the position $p$ of the ship has been determined (they are not necessarily obtained after a hit). Since each leaf corresponds to a different position $p \in S$ of the ship, it follows that there are exactly $n$ leaves.

The efficiency of the algorithm depends on the number of misses in the path going from the root to a leaf corresponding to position $p$. This number of misses for a position $p$ using tree $T$ is denoted $\m_T(p)$. 
We omit the subscript $T$ in $\m_T(S)$ when the decision tree $T$ is clear from the context.

The \emph{complexity of the algorithm} $T$ is the maximum number of misses in a path going from the root to a leaf, that is $\max_{p \in S}m_T(p)$. Note that the complexity is generally not equal to the height of the tree. 

Given a shape $S$, we define the Battleship \emph{complexity} $\comp(S)$ as the worst-case complexity considering all the decision trees $T$ that determine the position a ship of shape $S$:
\[\comp(S)=\min_T \, \max_{p \in S} \, \m_T(p).\] 

Next, we reuse the simple example of a horizontal line segment of length $4$, to illustrate our notation and framework. 
In this case, the shape $S$ is the set of lattice points \[S=\{ (0,0), (1,0), (2,0), (3,0) \}.\] 
An optimal algorithm for this shape has already been presented in the Introduction: after the initial shot at $x=(0,0)$, we shoot at values $x=(1,0),(2,0),\ldots$ until a miss occurs. This algorithm is modeled by the decision tree represented in Figure~\ref{f:simpletree}. The number of misses $\m(x)$ of this algorithm is equal to $0$ if $p = (0,0)$ and $1$ otherwise, giving a maximum of $1$, which proves $\comp(S) \leq 1$. It is easy to see that for any shape $S$ with $|S| > 1$, $\comp(S) \geq 1$. Hence, the algorithm is optimal.

\begin{figure}
  \begin{center}
		\includegraphics[scale=.75]{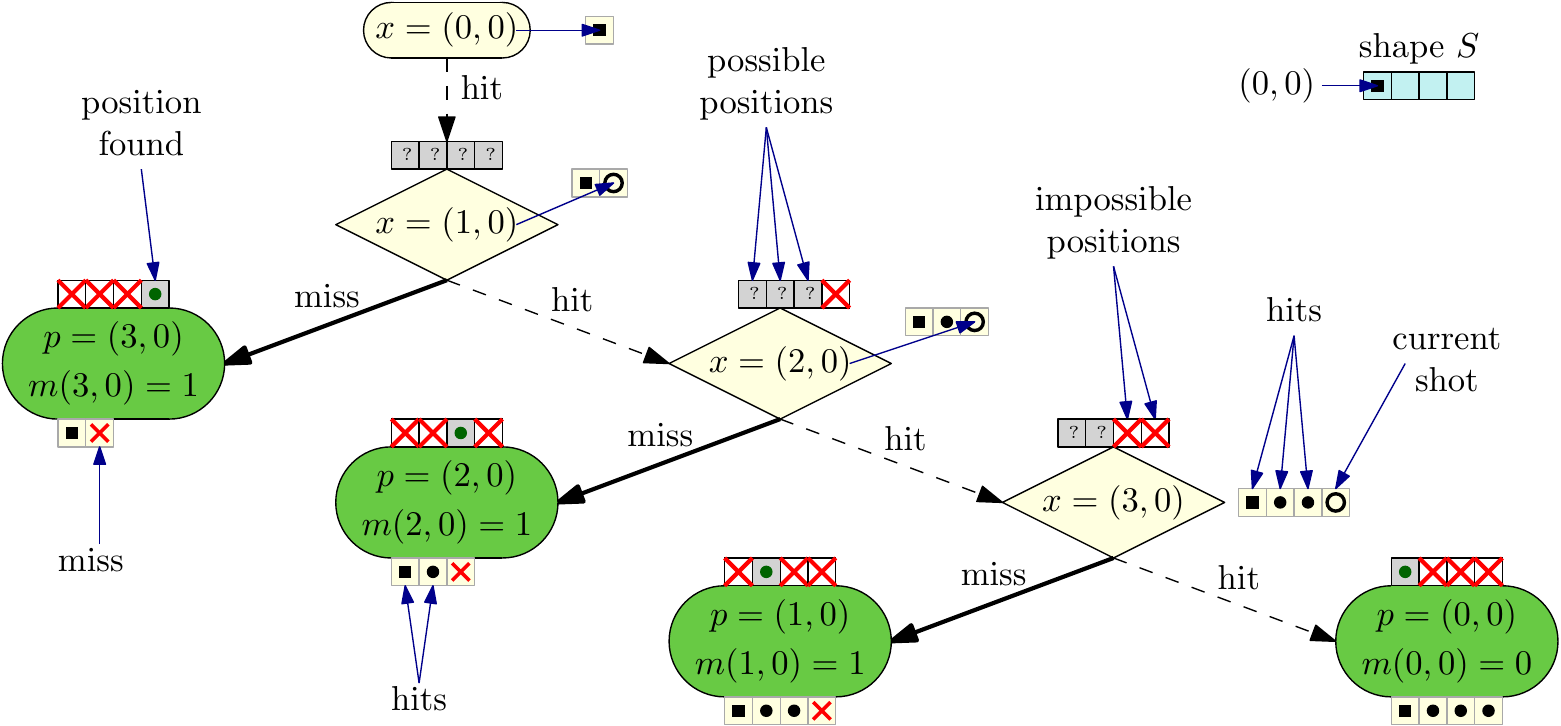}
	\end{center}
	\caption{\label{f:simpletree}Decision tree modeling an algorithm to sink the ship of horizontal shape $S=\{ (0,0), (1,0), (2,0), (3,0) \}$. The nodes correspond to the results of each. At each node, the set $P$ of the possible positions is represented by the gray squares in the small grid. 
	The leaves are the nodes where the ship position has been determined. The worst-case number of misses is $1$.
}
\end{figure}

A decision tree for a more complex shape is presented in Figure~\ref{f:simpletree2}. At each node of the tree, let $P$ be the corresponding set of possible positions. The set $P$ is represented by gray squares in the figure. The inclusion $p \in S$ is the only information that we have about the position of the ship when we start the algorithm (at the root of the decision tree). Hence, the set of possible positions at the root is $P=S$ and $P$ gets smaller at each new shot until it is reduced to a singleton at the leaves of the tree. 
At each new shot, $P$ is reduced in the following way (we use $S-x$ to denote a translation of the set $S$ by vector $x$):
\begin{itemize}
    \item If $x$ is a hit, then the set of possible positions for the child becomes $P \gets P \cap (S - x)$.
    \item If $x$ is a miss, then the set of possible positions for the child becomes $P \gets  P \setminus (S - x)$.
\end{itemize}

\begin{figure}
  \begin{center}
        \includegraphics[scale=.75]{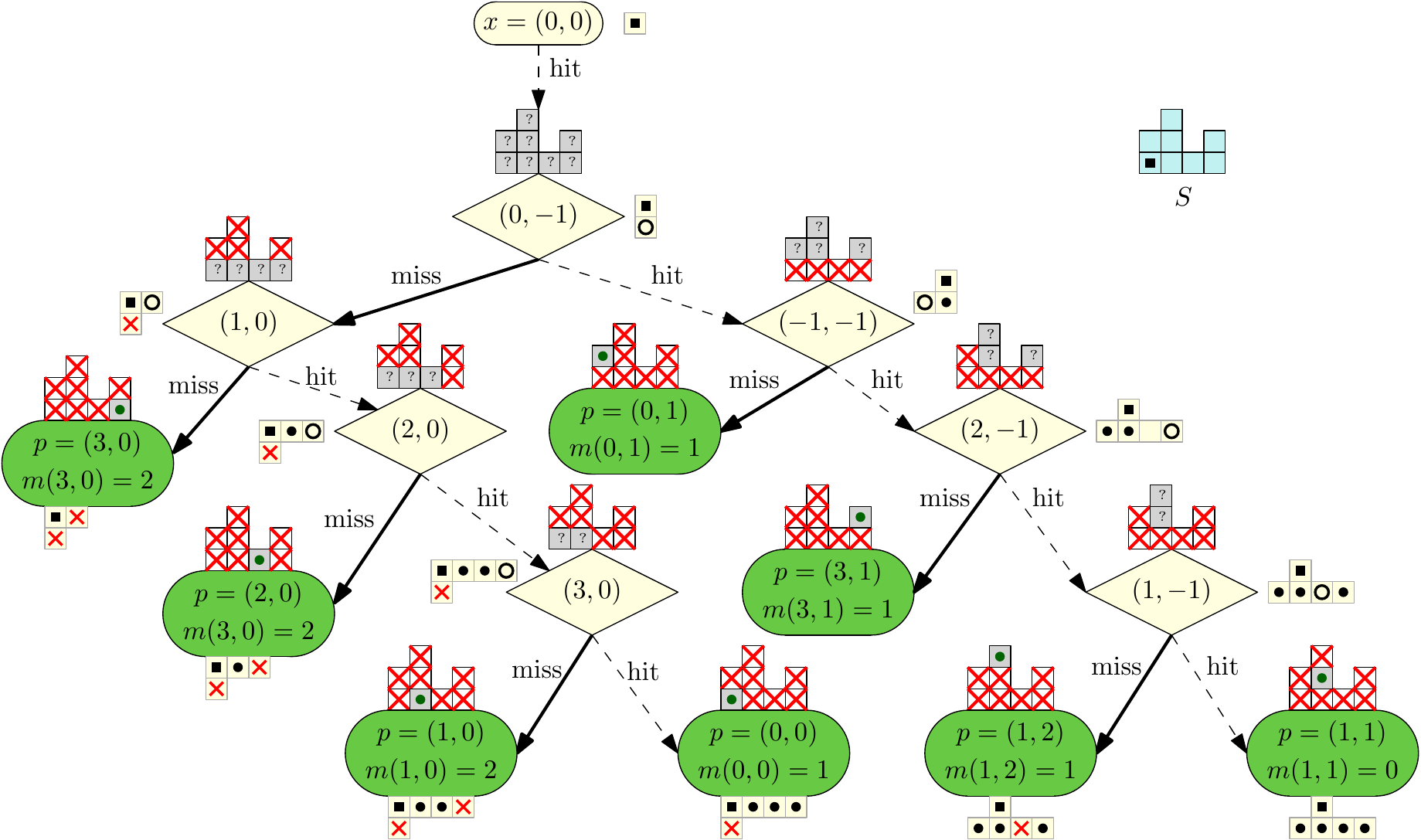}
	\end{center}
	\caption{\label{f:simpletree2}\textbf{A shape $S\subset \Z ^2$ and an algorithm to sink the ship having this shape with at most $2$ misses}. We follow the same graphic code as in Figure~\ref{f:simpletree}. This algorithm $T$ has a worst case complexity $\m(T)=2$. 
	The worst-case complexity of this shape $S$ is exactly $\m(S)=2$.
}
\end{figure}

It is easy to see that whatever the set of possible positions $P$ is, there always exists a shot which allows us to split $P$ in two non-empty subsets $P \setminus (S - x)$ and $ P \cap (S - x)$. Hence, the number of elements in $P$ strictly decreases as we move from a parent to a child and $\m(S) \leq n-1$ for all shapes $S$.

\subsection{Connection with Classification Trees} \label{ss:classification}

Classification trees are decision trees involved for instance in data mining for identifying an element $p$ belonging to a discrete set called the \textit{source set} and denoted $S$~\cite{Rok14}. The element $p$ is identified through the outcomes of a sequence of tests, where the choice of the new test depends on the previous outcomes. This dependency is modeled by a tree whose root represents the initial test. More generally, any internal node is associated to a test $T$ while its children correspond to the possible outcomes of $T$. 
Given an unknown element $p$, the algorithm to identify $p$ starts from the root. At each node, it considers the associated test and goes to the children node corresponding to the outcome of the test. The algorithm stops when arriving at a leaf: the leaf provides the identity of the unknown element $p$. 

The number of tests required to identify $p$ is the level of the corresponding leaf. Then the design of decision trees of small height is a well studied problem. This problem is known to be NP-hard in general~\cite{Hya76} (for minimizing the expected level of the leaves). The algorithm that chooses the most balanced test at each node provides an $O(\log n)$-approximation algorithm and the problem admits no polynomial $o(\log(n))$-approximation algorithms~\cite{Adl08}. 

Decision trees have been used in computational geometry for different purposes, for instance determining geometric models~\cite{Ark93a} or concept classes~\cite{Ark93b} in an image or more recently for the $k$-sum problem~\cite{Kan19}. As far as we know, the question of designing efficient strategies for playing Battleship with different types of shapes has not been addressed.

However, our problem possesses a fundamental difference in comparison to the classical use of classification trees: the Battleship problem is not symmetric. The goal in Battleship is to sink the ship with the minimal number of shots, but, since the number of hits to sink a ship is always equal to $n$, the goal becomes to minimize the number of misses needed to locate the position of the ship. 

The simplest heuristic to design classification trees of small height is to choose at each node the most balanced test possible. In our case, that would mean to choose a test such that the number of elements in $P \cap (S - x)$ and $P \setminus (S - x)$ are as similar as possible. As our goal is to minimize the number of misses instead of the height, we believe that a good heuristic strategy is to choose a shot such the number $|P \setminus (S - x))|$ of possible positions in case of a miss is as small as possible. However, we have not been able to prove any good worst-case bounds for this heuristic.

\subsection{Parallelogram-Free Shapes}

We say that a set $S$ is \emph{parallelogram-free} if every pair of distinct points define a unique difference vector (see Figure~\ref{f:par_free} for an example).
In other words, $S$ does not contain two distinct pairs of distinct points $s_1 \ne s_2$, $s_3 \ne s_4$ such that $s_2-s_1 = s_4-s_3$. 
In this section, we show that if $S$ is parallelogram-free, then $\comp(S) = n-1$. 
\begin{figure}[ht!]
  \begin{center}
		\includegraphics[scale=.75]{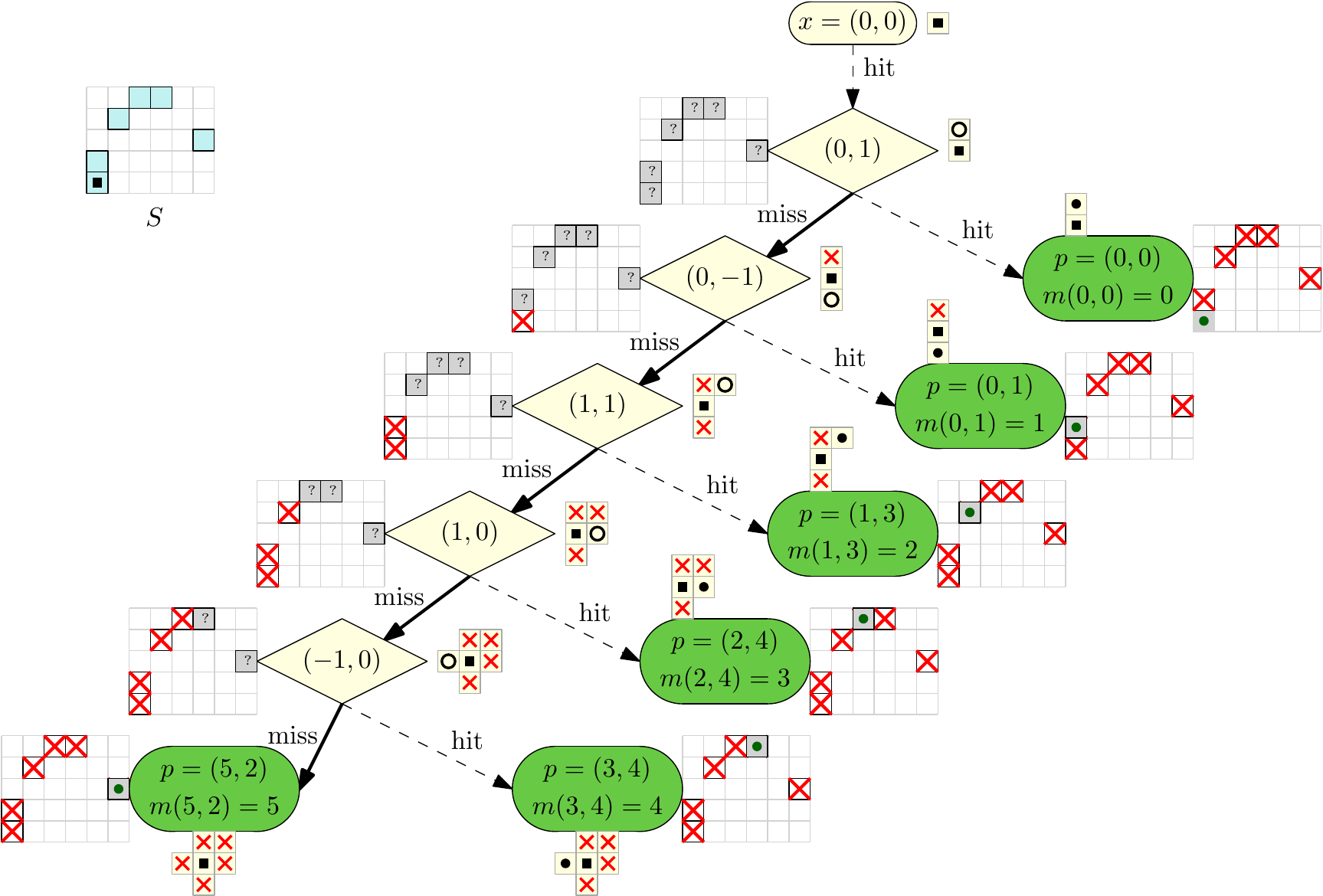}
	\end{center}
	\caption{\label{f:par_free}\textbf{A parallelogram-free shape $S\subset \Z ^2$}. In this case, all algorithms have isomorphic decision trees.}
\end{figure}

\begin{theorem}
If $S$ is a parallelogram-free polyomino of $n$ points, then $\comp(S) = n-1$.
\end{theorem}

\begin{proof}
We have already showed that $\comp(S) \leq n-1$ for any shape. Next, we show that $\comp(S) \geq n-1$ for a parallelogram-free shape $S$. To do this, we show that whenever we obtain a hit in the tree, we have successfully determined the position $p$. Hence, the miss branch of the tree contains all remaining points.

At any node of the tree, the new set of positions after a hit by shooting $x \ne (0,0)$ is defined by $P \gets P \cap (S - x)$.
Let us assume to obtain a contradiction that $y$ and $y'$ are distinct points in $S \cap (S-x)$ , then $y = s-x$ and $y' = s'-x$ with again $s$ and $s'$ both in $S$. It follows that  $x=y-s=y'-s'$. As $x$ is not $(0,0)$ and due to the parallelogram-free property,  $y=y'$ which contradicts the assumption.
\end{proof}

\section{HV-Convex Polyominoes} \label{s:hv}

In this Section, we investigate the Battleship complexity for the class of lattice sets of $\Z^2$ which are $4$-connected and HV-convex. 

\subsection{Definition and Properties}

A \emph{$4$-connected path} is a sequence $(x_1,\ldots,x_k)$ of distinct points of $\Z ^2$ such that the Euclidean distance between $x_{i}$ and $x_{i+1}$ is equal to $1$ for $i = 1,\ldots,k-1$.
A lattice set $S\subset \Z^2$ is \emph{$4$-connected} if for any pair of points $x_1,x_k \in S$, there exists at least one $4$-connected path $(x_1,\ldots,x_k)$ in $S$. The $4$-connected finite lattice sets are called \emph{polyominoes} (Figure~\ref{f:hv_poly}).

HV-convexity is a notion of directional convexity. A lattice set  $S\subset \Z^2$ is horizontally (vertically) convex if the intersection of $S$ with any row (column) is a set of consecutive points.
A lattice set which is horizontally and vertically convex is said to be \emph{HV-convex} (Figure~\ref{f:hv_poly}).

\begin{figure}[ht!]
  \begin{center}
		\includegraphics[scale=.75]{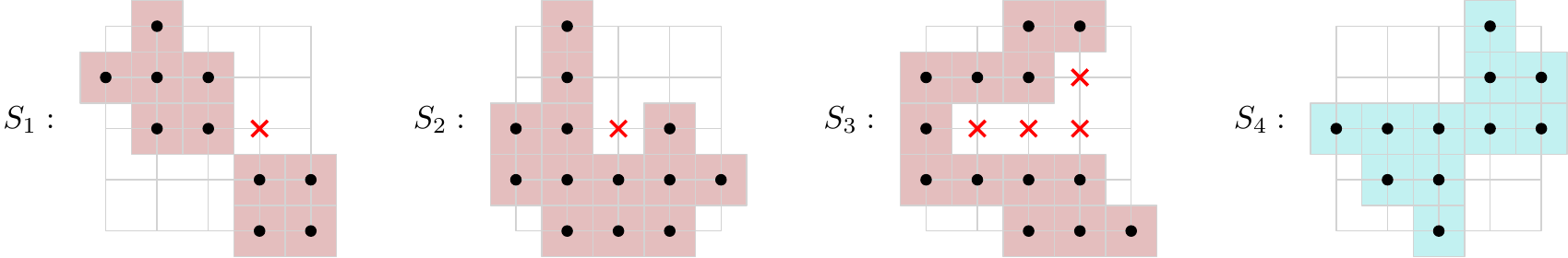}
	\end{center}
	\caption{\label{f:hv_poly}\textbf{Definition of HV-convex polyominoes}. The set $S_1$ is not $4$-connected, hence not a polyomino. The set $S_2$ a polyomino but not horizontally convex. The set $S_3$ is a polyomino but not vertically convex. The set $S_4$ is an HV-convex polyomino since it is $4$-connected, horizontally and vertically convex.
}
\end{figure}

We state two properties used in the following for proving the $O(\log n)$ bound on the complexity of HV-convex polyominoes.

\begin{lemma}\label{l:hv_stair}
Let $S$ be an HV-convex polyomino. If $(x,y)$ and $(x',y')$ are two different points of $S$ with $x \leq x'$ and $y \leq y'$, then either $(x+1,y)$ or $(x,y+1)$ is in $S$. 
\end{lemma}
\begin{proof}
If $x=x'$ or $y=y'$, the result is a direct consequence of the horizontal vertical convexities.
We consider now the case $x<x'$ and $y<y'$.
Let us assume that neither $(x+1,y)$ nor $(x,y+1)$ are in $S$, to obtain a contradiction. It follows that the vertical ray above $(x,y)$ and the horizontal ray to the right of $(x,y)$ do not contain any point of $S$. Then there is no $4$-connected path to connect  $(x,y)$ and $(x',y')$.
\end{proof}

After this general lemma about HV-convex polyominoes, let us introduce more specific material for our purpose, where we consider only the rows of fixed length $\ell$ (Figure~\ref{f:mono}): given an HV-convex polyomino $S$ and a fixed length $\ell\in \Z^+$, let $L$ be the number of rows of length $\ell$.
For $i$ from $1$ to $L$, we denote by $r_i$ the right endpoints of the $i$-th row of length $\ell$ ordered by $y$ coordinate. It follows that for all $i$, we have (i) $r_i \in S$, (ii) $r_i - (\ell,0) \not \in S$, (iii) $r_i - (\ell-1,0) \in S$, and (iv) $r_i + (1,0) \not \in S$.

\begin{lemma}\label{l:hv_monotonic}
 Given an HV-convex polyomino $S$ and a fixed length $\ell \in \Z$. The $x$-coordinate $x_i$ of the right endpoints $r_i$ of the rows of length $\ell$ forms a monotonic sequence (either  $x_i \leq x _{i+1}$ for all $i\in \{1,\ldots,L-1\}$ or $x_i \geq x _{i+1}$ for all $i \in \{1,\ldots,L-1\}$, as in Figure~\ref{f:mono}(b)).
\end{lemma}

\begin{proof}
If the sequence of $x$-coordinates $x_i$ is not monotonic, then there exists a triplet of indices $i$, $i'$ and $i''$ with $i<i'<i''$ leading to a configuration which is neither $x_{i} \leq x _{i'} \leq  x_{i''}$ nor $x_{i''} \leq x _{i'} \leq  x_{i}$. There are $4$ remaining permutations that cannot happen in an HV-convex polyomino. We show that it is not possible to have $x _{i'} < x_{i} \leq x_{i''}$ (see Figure~\ref{f:mono}(c)), the other $3$ cases being analogous. Suppose it is the case in order to reach a contradiction.

By definition, there is no point in $S$ to the right of $r_{i'}$.
Hence every 4-connected path from $r_i$ to $r_{i'}$ intersects the vertical ray going down from $p' =  (x_{i'} + 1, 0)$. Let $p$ be a point in this intersection. Similarly, let $p''$ be a point in the intersection of the path connecting $r_{i''}$ to $r_{i'}$ that is in the ray going up from $p'$. All $p$, $p'$, and $p''$ have the same $y$ coordinate, but $p' \not \in S$ while $p,p''$ are in $S$, which contradicts vertical convexity.
\end{proof}

\begin{figure}[ht!]
  \begin{center}
    	\includegraphics[scale=.75]{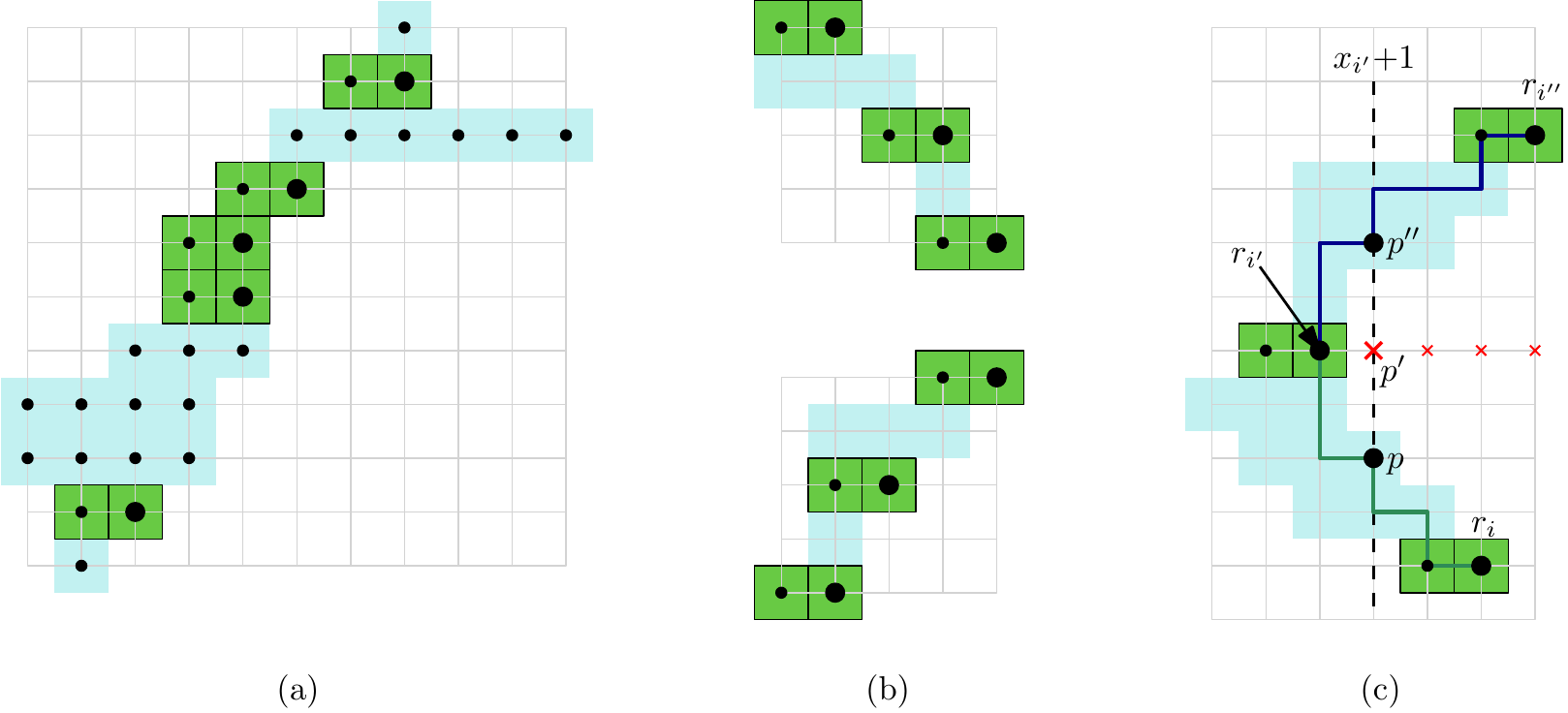}
	\end{center}
	\caption{\label{f:mono}\textbf{Monotonicity of the right endpoints coordinates} for the rows of fixed length $\ell$ of an HV-convex polyomino (here, $\ell=2$). (a) Lemma~\ref{l:hv_monotonic} states that for any length, the sequence of the $x$-coordinates of the right endpoints of the rows of a fixed length $\ell$ is monotonic (either increasing or decreasing). (b) Compatible configurations. (c) Non-monoticity is not compatible with the HV-convexity of a polyomino.
}
\end{figure}

\subsection{Shooting Algorithm with \texorpdfstring{$O(\log n)$}{O(log n)} Misses}

The previous lemmas allow us to develop an efficient algorithm for shapes that are HV-convex polyominoes.

\begin{theorem}\label{t:HV}
For any HV-convex polyomino $S$ of $n$ points, the Battleship complexity $\comp(S) = O(\log n )$.
\end{theorem}

We notice that the result does not hold for either HV-convex lattice sets (HV-convexity alone does not forbid arbitrarily large parallelogram-free lattice sets), or for arbitrary polyominoes (we let the reader construct counter-examples as a tricky exercise).

We prove the $O(\log n )$ bound of Theorem~\ref{t:HV} by providing a shooting algorithm with at most $O(\log n )$ misses for locating the position of the ship. We call it the \emph{staircase shooting algorithm}.

\subsubsection*{The Staircase Shooting Algorithm}

The staircase shooting algorithm for HV-convex polyominoes works in two phases.
The first phase consists of two sequences of horizontal shots going away from the origin in both horizontal directions. First we shoot at $(k,0)$ with an increasing $k = 1, \ldots, k^+$ until we obtain a miss at $k=k^+$. We proceed similarly in the negative direction until we obtain a miss at $k=k^-$. This way, we determine two values $k^- < 0$ and $k^+ > 0$ such that all the shots $(k,0)$ with $k^- < k < k^+$ are hits while $(k^-,0)$ and $(k^+,0)$ are both misses. The difference $k^+ - k^- + 1$ provides the length $\ell$ of the row of $S$ containing the unknown position $p$. We used $2$ misses to obtain the value of $\ell$ and concluded the first phase.

We now present the decision tree of the remainder of our algorithm. At any node, the set $P$ corresponds to the set of possible positions.
The number of rows of length $\ell$ in $S$ is denoted $L$. After the first phase, we know that the unknown position $p$ belongs to one of these $L$ rows of length $\ell$ and we know the horizontal position is the $k^-$-th point of the row. Hence, at this point, we have $|P| = L$ with at most one position in $P$ for each row. The positions are denoted $p_i$ for $1 \leq i \leq |P|$, ordered by $y$-coordinates.
It follows from Lemma~\ref{l:hv_monotonic} that the sequence of the $x$-coordinates of $p_i$ is either increasing or decreasing. We assume without loss of generality that the sequence is increasing, the other case being analogous.
The problem is to further reduce the set $P$ of possible positions. At any point, if $|P|$ is at most $2$, then we distinguish the $2$ possible positions with only $1$ additional miss.

\begin{figure}[ht!]
  \begin{center}
		\includegraphics[scale=.75]{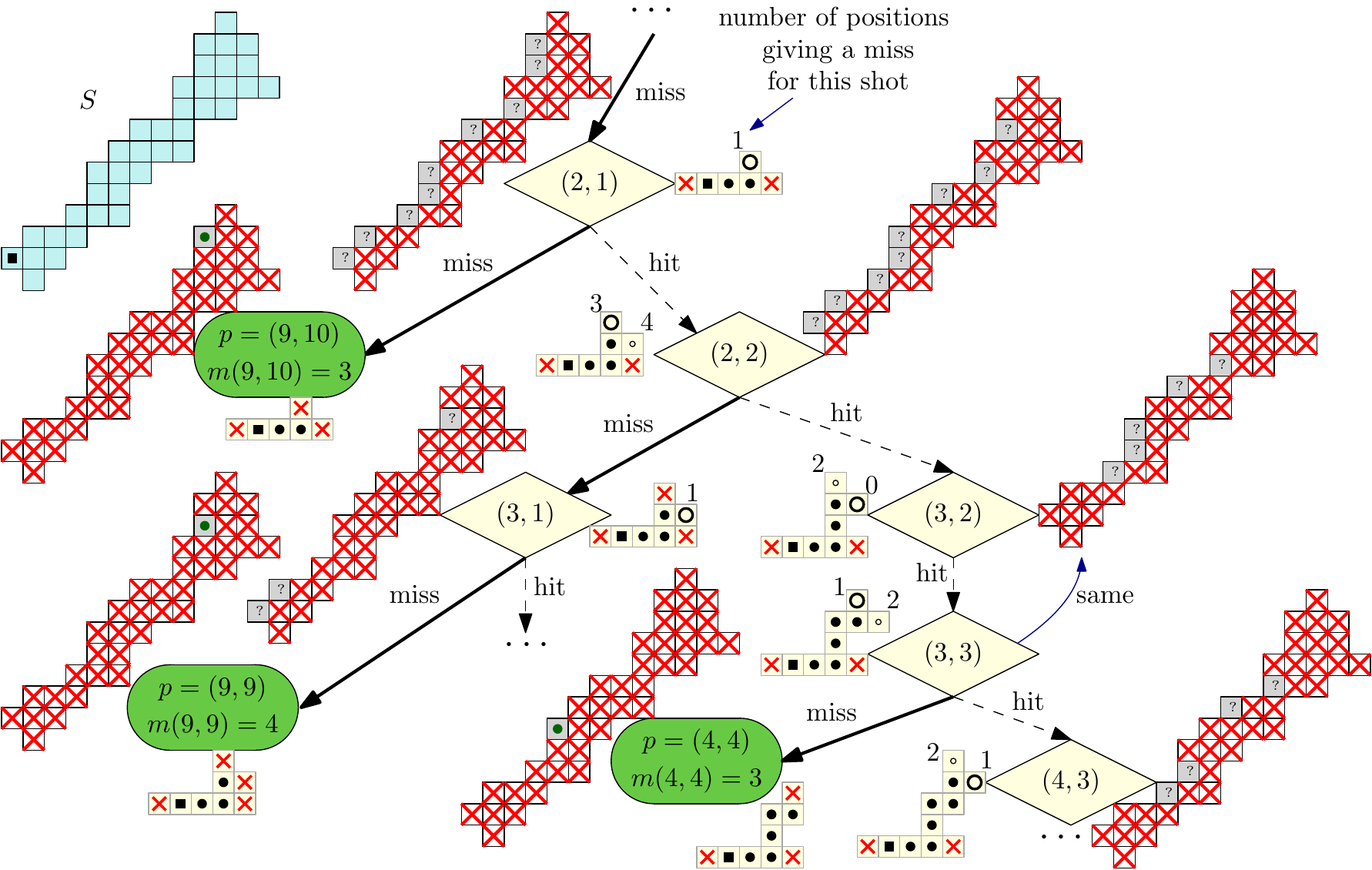}
	\end{center}
	\caption{\label{f:staircase}\textbf{Second phase of the staircase algorithm} on an HV-convex polyomino. Only part of the tree is represented.
}
\end{figure}

We now describe the second phase, assuming $|P|\geq 3$.
The sequence of hits follows a monotone $4$-connected path $(s_1,\ldots,s_J)$ of $J$ points of $S$ with $s_0 = (k^+-1,0)$ and shaped as a staircase going up and to the right, that is either $s_{j+1}=s_j + (1,0)$ or $s_{j+1}=s_j + (0,1)$.
For each $j$, we have to choose $s_{j+1}$ among the two possibilities. We proceed as in the heuristic described in Section~\ref{ss:classification}, choosing to shoot at the position $x$ that minimizes the number of positions $p_i \in P$ for which we would have a miss and let $x'$ denote the other choice. We now consider the two possible outcomes after shooting at $x$.

\begin{itemize}
    \item The child node after a \emph{hit}: At this new node, we have a new set of positions $P$ (which may or may not have been reduced), an unchanged number of misses, and a new node $s_{j+1} = x$ appended to the path. 
    
    \item The child node after a \emph{miss}: 
    In this case, the number of misses increased by $1$ and we cannot append $x$ to the path, since $x \not \in S$. We proceed with another shot at position $x'$ set to the other possibility to build the staircase path. If $x'$ is also a miss, then we determined that the position $p$ is the top-most point of $P$ namely $p = p_{|P|}$  (Claim~\ref{c:topmost}). 
    However, if $x'$ is a hit, then we append $s_{j+1}=x'$ to the staircase path. In this case, we will show that at least $1/3$ of possible positions $P$ have been discarded (Claim~\ref{c:bound}).     
\end{itemize}

\subsubsection*{Proof of the Claims}

We consider the conditions of the algorithm at a current node associated to a monotonous staircase path  $(s_1,\ldots,s_J)$ with $s_1 = (k^+ - 1, 0)$ such that all shots $s_j$ for $j = 1,\ldots,J$ provided hits. 
The set of the possible positions is a set of $|P|$ points $p_1,\dots,p_{|P|}$ at the $k^-$-th position in rows of length $\ell$.
The path $s_1,\ldots,s_{J}$ being all hits, we know that for any possible position $p_i$ and any shot $s_j$ of the path, the sum $p_i+s_j$ is a point of $S$.
We remind the reader that the points $p_i$ are ordered by $y$-coordinates and that we have assumed without loss of generality that their $x$-coordinates are increasing.  

\begin{claim}\label{c:topmost}
We consider the two new possible shots $s_j + (1,0)$ or $s_j + (0,1)$.
For all possible positions $p_i$ with $i < |P|$, one of the two shots provides a hit. In other words, the unique possible position for which we might obtain two misses is the top-most $p_{|P|}$.  
\end{claim}

\begin{proof}
We have to prove that for any possible position $p_i$ 
with $i < |P|$, then either $s_j + (1,0)$ or $s_j + (0,1)$ provides a hit. It means that 
either $p_i + s_j + (1,0) \in S$ or $p_i + s_j + (0,1) \in S$. 
This property is a direct consequence of Lemma~\ref{l:hv_stair}, because $p_i + s_j \in S$, $p_{|P|} + s_j \in S$, and
$p_{|P|}+s_j $ is in the northeast quadrant of $p_i+s_j$ (namely the $x$ and $y$-coordinates of $p_i +s_j$ are respectively lower than the ones of $p_{|P|} + s_j$).
\end{proof}

Claim~\ref{c:topmost} leads to a second claim under the assumption that $|P| \geq 3$.

\begin{claim}\label{c:bound}
At each current node with $|P| \geq 3$, by choosing 
between $s_j + (1,0)$ or $s_j + (0,1)$  the shot for which the number of possible positions $(p_i)_{1\leq i \leq I}$ remains the largest, the number of positions providing a miss is at most $\frac{2}{3} |P|$. 
\end{claim}

\begin{proof}
Let $P_{(1,0)}$ be the set of positions $p_i$ such that $p_i + s_J + (1,0) \in S$ and let $P_{(0,1)}$ be the set of of positions $p_i$ such that $p_i + s_J + (0,1) \in S$. Let $n_{(1,0)} = |P_{(1,0)}|$ and $n_{(0,1)} = |P_{(0,1)}$.
According to the claim~\ref{c:topmost}, with the exception of the topmost possible position $p_{|P|}$, all the others give a hit for one of the two possible shots.
Hence, we have
$n_{(1,0)} + n_{(0,1)} \geq |P| - 1$ and
\[\max(n_{(1,0)}, n_{(0,1)}) \geq \frac{|P| - 1}{2} \geq \frac{|P|}{3},\]
since $|P| \geq 3$.
Then with the shot $s_J + (1,0)$ or $s_J + (0,1)$ corresponding to the maximum of $n_{(1,0)}, n_{(0,1)}$, we have at least $1/3$ of the possible positions giving a hit.
\end{proof}

\subsubsection*{Complexity Analysis}

We need to count the number of misses that the staircase algorithm takes in the worst case.
The first phase of the staircase algorithm uses $2$ misses. In the second phase, until the number of possible positions falls under $3$, each miss reduces the number of possible positions by a factor of at most $2/3$ of the previous value. Since the initial value of $|P|$ is at most $n$, the number of misses during this phase is at most $k$ where $k$ is the smallest integer verifying  $(2/3)^k n \leq 3$. It follows that $k = O(\log n)$. In the last step, we finish the computation with at most $2$ more misses.

\section{Digital Convex Sets} \label{s:convex}

A shape $S \subset \Z^2$ is \emph{digital convex} if there exists a convex polygon $K \subset \R^2$ such that $S = K \cap \Z^2$.
Digital convexity is a stronger property than HV-convexity, however it does not imply $4$-connectivity.
We can test if a set $S \subset \Z^2$ is digital convex by verifying if $\mathrm{conv}(S) \cap \Z^2 = S$ (Figure~\ref{f:dc}), where $\mathrm{conv}(S)$ denotes the continuous convex hull of $S$. This property is exploited in~\cite{CFG19} to test digital convexity in linear time.

\begin{figure}[t]
  \begin{center}
		\includegraphics[scale=.75]{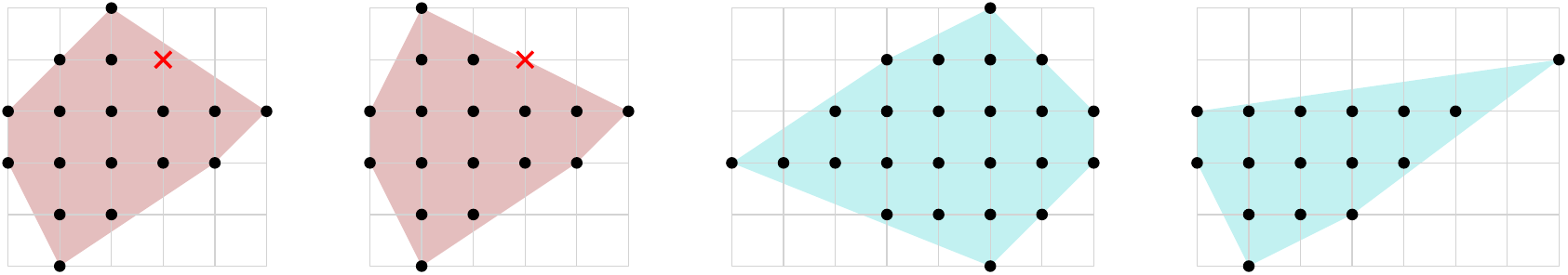}
	\end{center}
	\caption{\label{f:dc}\textbf{Digital convexity.} The two lattice sets on the left are not digital convex while the two on the right are, since there is no other lattice point in their convex hulls.}
\end{figure}

In this section, we investigate the Battleship complexity of digital convex sets. We choose this family of lattice sets not only because it is one of the main classes of geometric shapes but also because in a continuous variant of the problem, the Battleship complexity of the convex sets is bounded by a constant. Determining if the same holds for the digital version is an intriguing question.
In Section~\ref{ss:digitalconvex}, we prove a new version of the Blaschke-Lebesgue inequality and in Section~\ref{ss:loglog}, we present an $O(\log\log n)$ algorithm.

\subsection{Discrete Blaschke-Lebesgue Inequality}\label{ss:digitalconvex}

Digital convex sets have inequalities relating the lattice diameter and the lattice width. First, we recall their respective definitions. Let $S$ be a digital convex set. Its \emph{lattice diameter} $d(S)$ is the maximum number of points of $S$ on a line, minus $1$ (Figure~\ref{f:diameter}). It follows that the lattice diameter of a single point is $0$.
A \emph{Diophantine line} is a line containing at least two lattice points. Two Diophantine lines are \emph{consecutive} if they are parallel to each other and there is no lattice point between them. The \emph{lattice width} $w(S)$ of $S$ is the minimum number of consecutive Diophantine lines covering $S$, minus $1$ (Figure~\ref{f:diameter}). The lattice width of a single point is again $0$. More formally, the lattice width can be expressed as 
\[w(S)= \min _{u\in \Z ^2 \setminus \{(0,0)\}} \;\; \max_{a,b \in S} \; u \cdot (b-a). \]

\begin{figure}[hbt]
  \begin{center}
		\includegraphics[scale=.75]{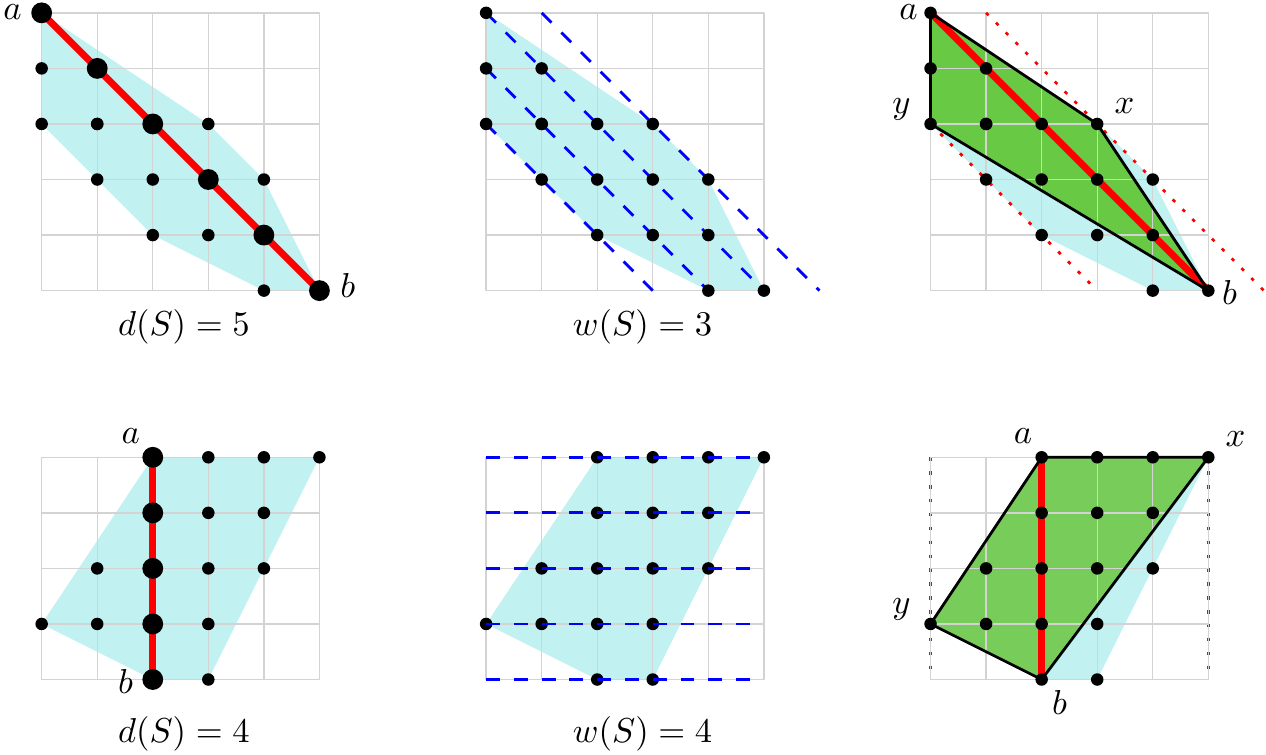}
	\end{center}
	\caption{\label{f:diameter}\textbf{Diameter and  lattice width} of two lattice sets. 
    On the right, we show the quadrilaterals $xayb$ used in the proof of Lemma~\ref{l:bla}.
}
\end{figure}

The \emph{continuous width} of a convex body $K$ is the minimum distance between two parallel lines enclosing $K$ and is denoted by $width(K)$.
The Blaschke-Lebesgue theorem states that the area of a convex body $K \subset \R^2$ is at least $\frac{\pi - \sqrt 3}{2} width(K)^2$. This lower bound is achieved when $K$ is the so called Reuleaux triangle.
I. Bárány and Z. Füredi~\cite{Bar01} provided the following discrete version of the theorem. 

\begin{lemma}\label{l:Barany}
For any digital convex set $S$, we have $w(S) \leq \lfloor \frac{4}{3} d(S) \rfloor + 1$ and for any fixed diameter, this bound is best possible.
\end{lemma}

This inequality is not exactly an equivalent of the Blaschke-Lebesgue theorem since a lower bound on the discrete diameter does not directly provide a lower bound on the area. It remains a small gap to fill in order to obtain a more standard equivalent of the Blaschke-Lebesgue theorem for digital convex sets.
We provide a new discrete inequality closer to the original Blaschke-Lebesgue theorem where the number of points of the lattice set $S$ plays the role of the area and the lattice width plays the role of the width. 

\begin{lemma}\label{l:bla}
For any digital convex set $S$ of $n$ points, we have $n \geq \frac{3}{8} w(S)^2 - \frac{1}{2} w(S) +3$.
\end{lemma}
\begin{proof}
Let us denote $a$ and $b$ the pair of extreme points providing the diameter $d(S)$. It follows $b-a=d(S)v$ where $v$ is the vector in the direction $b-a$ with coprime coordinates. Let $u$ be the rotation of $v$ by $\frac{\pi}{2}$.
We consider the points $x$ and $y$ of $S$ minimizing and maximizing the dot product with $u$. By definition of the lattice width, $u \cdot (y-x) \geq w(S)$ (i).
The four points $x$, $a$, $y$, and $b$ define the convex quadrilateral $xayb$ (Figure~\ref{f:diameter}). Its area is $A=\frac{1}{2} |\det(y-x,b-a)|=\frac{1}{2} |d(S)u\cdot(x-y)|$. 
With (i), we obtain $A \geq \frac{1}{2} d(S) w(S)$ (ii). 

Pick's theorem allows us to calculate the number of lattice points in the quadrilateral $xayb$ from its area. We recall the formula $A = i+\frac{e}{2}-1$ where $i$ is the number of interior points and $e$ the number of points on the boundary of $xayb$. By denoting $n_{xayb}$ the number of lattice points in the quadrilateral $xayb$, we have $n_{xayb} = i+e$. Then Pick's formula provides $n_ {xayb}=A+1+\frac{e}{2}$. The number of points on the boundary of the quadrilateral being at least $4$, we have 
$n_{xayb} \geq A+3$. With the bound (ii) on the area $A$, we obtain $n_{xayb}\geq \frac{1}{2} d(S)w(S) +3$. As the set $S$ is digital convex, it contains all the lattice points in the quadrilateral $xayb$: $n\geq n_{xayb}$. Then we have $n\geq \frac{1}{2} d(S)w(S) +3$ (iii).

Lemma~\ref{l:Barany} provides the bound $w(S)\leq \frac{4}{3}d(S)+1$ which can be rewritten $d(S)\geq \frac{3}{4} (w(S)-1)$. With (iii), it gives $n \geq \frac{3}{8} w(S)^2 - \frac{1}{2} w(S) +3$.
\end{proof}

We can write a similar bound which is easier to use as follows. 

\begin{lemma}\label{l:bla2}
For any digital convex set $S$ of $n$ points, we have $n \geq \frac{1}{4} w(S)^2 $.
\end{lemma}
\begin{proof}
We have $\frac{1}{4} x^2  \leq \frac{3}{8} x^2 - \frac{1}{2} x +3$ for any real $x$. Then we can rewrite Lemma~\ref{l:bla} with $\frac{1}{4} w(S)^2 $ as new lower bound. We could even write $n \geq \frac{1}{4} w(S)^2 +2$.
\end{proof}

Lemma~\ref{l:bla2} shows that the lattice width is bounded by the square root of the number of points of a digital convex sets. This relation  is the key point for proving the complexity of the next algorithm. 

\subsection{Algorithm with \texorpdfstring{$O(\log \log n)$}{O(log log n)} Misses}\label{ss:loglog}

In this section, we prove the following theorem.

\begin{theorem}\label{t:final}
For any digital convex set $S$ of $n$ points, the Battleship complexity $\comp(S) = O(\log \log n )$.
\end{theorem}

We call this algorithm the \textit{width shooting algorithm} because it is mainly based on shots in the direction given by the lattice width of the set of possible positions. Next, we describe the algorithm.

Consider a node in the tree associated with a set $P$ of possible positions. Initially $P = S$, but as the algorithm progresses, positions will be removed from $P$.
We compute the lattice width $w(P)$, which is achieved by a vector $v \in \Z^2$ with coprime coordinates. Then, we rotate $v$ by $\frac{\pi}{2}$, obtaining a vector $u$ that is parallel to $w(P) + 1$ lines covering the set $P$.
Then, we shoot in directions $u$ and $-u$ with shots of the form $ku$ for positive and negative integer $k$ until we obtain misses at points $k^+u$ and $k^-u$ with hits in between.
These two misses and the previous hits lead to a new current node with a new set of possible positions that we denote $P'$. 

We repeat this procedure from node to node until we obtain a set of possible positions whose convex hull has fewer than $25$ points. When we reach this value, then any shooting algorithm can be used with at most $24$ misses.  
Theorem~\ref{t:final} follows from the claim that this algorithm has $O(\log \log n)$ misses. 

\subsubsection*{Convex Hull of the Possible Positions}
 
The main procedure of the width shooting algorithm uses shots $ku$ with $k\in \Z$ in direction $u$ until finding the two boundary points of the ship in this direction. 
After this sequence of shots, due to the digital convexity of $S$, the new set of possible positions $P'$ has the following property.

\begin{claim}\label{c:dio}
The convex hull of $P'$ contains at most one point on each Diophantine line parallel to $u$. 
\end{claim}

\begin{proof}
Let $k^+ \in \Z ^+$ and $k^- \in \Z ^-$ be the first positive and negative integers for which the shots $ku$ give a miss (starting from $k=0$).  
The difference $k^+ - k^- -1$ is equal to the length $\ell$ of the intersection of the shape $S$ and the Diophantine line $p + k u$ for $k \in \Z$. 
After the two misses obtained with the shots $k^+ u$ and $k^- u$, the set $P'$ of the possible positions satisfies
(i)  $( P' + k^+ u ) \cap S = \emptyset$ and
(ii) $( P' + k^- u ) \cap S = \emptyset$.

As the previous shots are hits, we also have 
$P' + (k^+ - 1) u \subset S$ and $P' + (k^- + 1) u \subset S$. According to the digital convexity of $S$, it follows from the two last inclusion that the convex hulls of these two sets are still included in $S$: 
$\conv(P' + (k^+ - 1) u) \cap \Z ^2 \subset S$ and $\conv(P' + (k^- + 1) u) \cap \Z ^2 \subset S$. With (i) and (ii), we obtain 
$( P' + k^+ u ) \cap \conv(P' + (k^+ - 1) u) = \emptyset$
and 
$( P' + k^- u ) \cap \conv(P' + (k^- + 1) u) = \emptyset$. With translations, it leads to (iii) $( P' + u ) \cap \conv(P') = \emptyset$ and (iv) $( P' - u ) \cap \conv(P') = \emptyset$.

To arrive at a contradiction, we assume that there exists a pair of distinct points in the convex hull of $P'$ and on the same Diophantine line parallel to $u$. 
In~\cite{Bar01, Bar92}, they show that a segment in direction $u$, included in the convex hull of $P'$ and of maximal length has a vertex of the convex hull $a \in P'$ as an endpoint.
Hence, we have that $a$ is a vertex of the convex hull of $P'$ and as the maximal length is at least $\|u\|$, either $a+u$ or $a-u$ is in the convex hull of $P'$. It contradicts either (iii) or (iv).
\end{proof}

\subsubsection*{Complexity Analysis}

We prove Theorem~\ref{t:final} by computing the worst case number of misses of the width shooting algorithm. 
At the beginning, the set of the possible positions $P$ is initialized as $S$. Its convex hull contains no more than the $n = |S|$ points of $P$. 
After the first step, according to Claim~\ref{c:dio}, the convex hull of the new set of positions $P'$ has no more than one point per Diophantine line in the chosen direction $u$. It follows that the number of points of the convex hull of $P'$ is less or equal to the number of Diophantine lines covering $P$  which is $w(P)+1$ since we choose the direction $u$ providing this value.
We use now the bound of Lemma~\ref{l:bla2}: the inequality $n \geq w(S)^2 / 4$ leads to $w(S) + 1 \leq 2 \sqrt n + 1$. For $n$ larger than $25$, we  have $2 \sqrt n + 1 \leq n^{3/4}$. It means that except if the number of possible positions falls under $25$, the convex hull of the new set of possible positions contains fewer than $n^{3/4}$ lattice points.

By iterating $k$ times the procedure (each time uses $2$ misses), we have a set of possible positions whose convex hull contains at most 
$n^{(3/4)^k}$ points. The number of iterations $k$ to get to fewer than $25$ points is hence $O(\log \log n)$. Since each iteration has a constant number of misses, the total number of misses is also $O(\log \log n)$.

\section{Conclusion and Open Problems} \label{s:conclusion}

A simplified version of the children's game Battleship leads to numerous nontrivial questions and algorithms that we had a lot of fun to work on. We worked on the digital version of the problem, which is directly connected to the actual game. However, a continuous variation may also raise interesting questions.

Let $S \subset \R^2$ be a convex body. In the continuous version, instead of querying a point, the player can shoot along a ray until finding a miss (ray-shooting queries). The information that the player gets is the position of the boundary point of $S$ on the ray.
The problem consists of recovering the unknown position $p$ of $S$ with a small number of shots and it can be solved in the following manner. Let $v$ denote the direction of a diameter of $S$.
We shoot in directions $v$ and $-v$ from the origin and determine with two queries the length of the line segment in direction $v$ that passes through the origin. Since the direction $v$ is a diameter direction, it is not possible to hit two parallel edges in the ray-shooting queries. Hence, there are at most two points that can give the same two results from the ray shooting queries. A third and last query is sufficient to distinguish these two points.

We conclude by listing several questions that remain open.

\begin{enumerate}
    \item Given a finite shape $S \subset \Z^2$, what is the complexity to actually calculate its Battleship complexity $\comp(S)$? Is the problem NP-complete as in the case of general minimal decision trees?
    
    \item We consider the greedy shooting heuristic choosing at each node the shot providing the minimal number of misses. Does the heuristic provide an $O(\log n )$ approximation to the minimum number of misses? Does there exist better approximation algorithms?
    
    \item For the class of polyominoes that are not HV-convex, we can build examples showing that the Battleship complexity is $\Omega(\log n )$. However, there is still a big gap with the upper bound of $n-1$. Can this gap be reduced?

    \item For the class of digital convex sets, we provide an algorithm with at most $O(\log \log n )$ misses but in practice, the largest Battleship complexity that we found with the heuristic is only $3$ as for continuous shapes. Is it possible that the Battleship complexity of the digital convex sets is also bounded by a constant?
    
    \item Are there some other interesting classes of lattices sets for which efficient shooting algorithms can be found?
    
    \item Lemma~\ref{l:bla} states the inequality $n\geq \frac{3}{8} w(S)^2 - \frac{1}{2} w(S) +3$ providing a lower bound on the number of points $n$ of a digital convex set $S$ according to its lattice width $w(S)$. This discrete version of Blaschke-Lebesgue inequality is however not tight. What is the best bound that can be achieved?
    
    \item We defined the Battleship complexity in terms of the maximum number of misses. We could define an average version of the complexity, in which the average of $\m(x)$ for $x \in S$ is considered instead. What can be said about the average complexity?
    
    \item If we count the total number of shots, instead of the number of misses we can adapt the algorithm from Section~\ref{ss:loglog} to obtain an algorithm for digital convex shapes that uses $O(\log n)$ shots and that is optimal. What is the complexity of this variation for HV-convex polyominoes?
    
    \item What is the complexity of the continuous version for different classes of shapes if both rotations and translations are allowed? Is it still constant for convex shapes? 
\end{enumerate}


\bibliography{ref}

\end{document}